\documentclass{llncs}

\usepackage{soulutf8} 
\usepackage{todonotes} 
\usepackage{cite} 
\usepackage{amssymb}
\usepackage{amsmath}
\usepackage{textcomp}
\usepackage{array}
\usepackage{color}
\usepackage{setspace}
\usepackage{comment}
\usepackage[]{datetime}
\usepackage{hyperref}
\usepackage{lineno}

\newtheorem{fact}{Fact}

\newcommand{\dspace}{\mathrm{SPACE}}

\newcommand{\nat}{\mathbb{N}}

\newcommand{\real}{\mathbb{R}}

\title{Computational Complexity of Space-Bounded Real Numbers}

\author{Masaki Nakanishi\inst{1} \and Marcos Villagra\inst{2}\thanks{corresponding author.}
}
\institute{
	Faculty of Education, Arts and Science, Yamagata University\\
	Yamagata, 990--8560, Japan
	\and
	National University of Asuncion\\
    	NIDTEC, Campus Universitario, San Lorenzo C.P. 2619, Paraguay\\
}

\authorrunning{M. Nakanishi \and M. Villagra} 

\begin{document}

\maketitle

\begin{abstract}
In this work we study the space complexity of computable real numbers represented by fast convergent Cauchy sequences. We show the existence of families of trascendental numbers which are logspace computable, as opposed to algebraic irrational numbers which seem to required linear space. We characterized the complexity of space-bounded real numbers by quantifying the space complexities of tally sets. The latter result introduces a technique to prove the space complexity of real numbers by studying its corresponding tally sets, which is arguably a more natural approach. Results of this work present a new approach to study real numbers whose transcendence is unknown.

\end{abstract}


\section{Introduction}

The set of computable real numbers form the basis of a modern theory of computable analysis. Real numbers were originally studied in computational complexity theory in the seminal work of Hartmanis and Stearns \cite{HS65}, where it was shown that algebraic numbers are polynomial-time computable. Later, Cobham \cite{Cob69} showed that no finite state automaton can compute the digits of algebraic irrational numbers. From those works stemmed what today is known as the \emph{Hartmanis-Stearns conjecture}, which states that if a real number is real time computable with respect to some natural base, then that number is either rational or trascendental. The Hartamnis-Stearns conjecture has many implications as recounted by Freivalds \cite{Fre12}, for example, it implies the no existence of optimal algorithms for integer multiplication and it will solve the transcendence of several real numbers, just to name a few. More recently, Adamczewski and Bugeaud \cite{AB07} made an important breakthrough towards the Hartmanis-Stearns conjecture where they showed that no algebraic irrational number is computable by a pushdown automaton. Currently, this is the only result where algebraic irrational numbers are not computable even in the presence of some non-constant amount of memory.

Heartmanis and Stearns \cite{HS65} presented an algorithm that computes any algebraic numbers in polynomial time and linear space. Motivated by this fact, in this work we study real numbers that are computable in bounded-space in order to understand if linear space is a necessary amount of memory for algebraic irrational numbers. First, in Theorem \ref{the:hierarchy} of Section \ref{sec:hierarchy} we show a space hierarchy theorem for real numbers, which implies that the space complexity of trascendental numbers cannot be bounded, unlike algebraic numbers. In Section \ref{sec:trascendental} we study the space complexity of trascendental numbers and we show general theorems for trascendental numbers with certain natural forms. Contrary to the situation of trascendental numbers where space-efficient algorithms exist, it is very difficult to construct space-efficient algorithms for algebraic irrational numbers.

In order to conduct a study in the theory of space-bounded real numbers, we follow an idea initiated by Ko \cite{Ko91} which relates the polynomial time computability of real numbers to classical computational complexity theory. Yu and Ko \cite{YK13} initiated a study of logspace computable real numbers where they showed how open problems in the theory of space-bounded tally sets relates to representations of a real number. In Theorem \ref{the:tally} we present a characterization of the space complexity of a real number in terms of the space complexity of a tally set representing the same number.
With this characterization, if we would like to design a space-efficient algorithm for a real number, it suffices to construct such an algorithm for its corresponding tally set, which is most of the time more natural to work with. In Theorem \ref{the:language} we present a relation between the set of left-cuts representations of a real number and its space complexity. As a result we have that tally sets corresponding to a real number and its left-cuts representations are polynomial-space equivalent. Finally, in Section \ref{sec:automata} we show that constant-space machines modeled by finite automata cannot recognize irrational numbers even with the help of advice.

For several cases of the real numbers studied in this work there exist space-efficient algorithms, whereas for algebraic irrational numbers the best algorithm needs to remember all digits previously computed. The results of this work thus suggest a new conjecture for algebraic irrational numbers, namely, that there is no algorithm that computes algebraic irrational numbers using sublinear space.




\section{Preliminaries on Computable Numbers and Complexity}
\subsection{Computable Numbers}
We use $\mathbb{N}$ to denote the set of the natural numbers including 0. In the rest of this paper, we will use multitape Turing machines with blank symbols denoted by $\boxtimes$.

A \emph{dyadic number} is a number of the form $d=m/2^n$ for integers $m$ and $n$ with $n\geq 0$. The binary expansion of a dyadic number is
\[
d=\sum_{i=0}^m s_i\cdot 2^i+ \sum_{j=1}^n t_j \cdot 2^{-j},
\]
where $s_i,t_j\in \{0,1\}$ and each $s_i$ and $t_j$ are the digits of the integer and binary parts of $d$, respectively. We say that the string $s=s_n\cdots s_0\cdot t_1\cdots t_m$ represents $d$ or $s$ is a representation of $d$ in base 2.

Given a representation $s$ of a dyadic number $d$, we denote by $\ell(d)$ the number of symbols in the representation $s$ of $d$ and we use $prec(d)$ to denote the number of symbols to the right of the binary point in $s$. Furthermore, we let $\mathbb{D}_n=\{d \> | \> prec(d)=n\}$ and $\mathbb{D}=\cup_{i=0}^\infty \mathbb{D}_n$. 

Given a function $\phi:\mathbb{N}\to \mathbb{D}$ we say that $\phi$ \emph{binary converges} to a real number $x$ if for all $n\in \mathbb{N}$, $|\phi(n)-x|\leq 2^{-n}$. Thus the sequence $\{\phi(n)\}_{n\geq 0}$ is a fast-convergence Cauchy sequence. If $\phi$ binary converges to $x$, we call $\phi$ a Cauchy function of $x$. We denote by $CF_x$ the set of all Cauchy functions of $x$; see Ko \cite{Ko91} for more details.

The function $\phi$ is computable if there exists a Turing machine that on input $n$ in unary outputs a representation of $\phi(n)$. A real number $x$ is \emph{computable} if there exists a computable function $\phi \in CF_x$. 

Given a function $T$ over $\nat$, the time (or space) complexity of a computable real number $x$ is bounded by $T$ if there exists a Turing machine that on input $0^n$ outputs a representation $s$ of a dyadic number $d$ such that $|d-x|\leq 2^{-n}$ and uses $T(n)$ moves (or $T(n)$ memory cells, respectively). As a short-hand we will use $time_x(n)$ and $space_x(n)$ to denote the time and space complexities of $x$, respectively. For any $f:\nat\to\nat$ we define the class $\dspace_\real(f(n))=\{x\in \real \> |\>  space_x(n)=\mathcal{O}(f(n))\}$. 

\begin{fact}\label{fac:rational}
	Any rational number can be computed in constant space.
\end{fact}


Hartmanis and Stearns \cite{HS65} present an algorithm for computing algebraic real numbers, which runs in linear space.

\begin{fact}\label{fac:algebraic}
	Any real algebraic number can be computed in linear space. (App.\ref{sec:algebraic})
\end{fact}

\subsection{Oracle Turing Machines and Reducibility}

Consider any Turing machine $M$. An oracle is a subrutine that is incorporated in $M$ and can be used to ask questions with answers ``yes'' and ``no.'' More formally, let $A$ be any set and let $M^A$ be a Turing machine that computes with $A$ as an oracle. The machine $M$ aditionally has one write-only query tape, one query inner state $q_{query}$ and two answer states, namely the ``yes'' state $q_{yes}$ and the ``no'' state $q_{no}$. The query tape works like an output tape, that is, it is write-only and everytime a symbol is written, the head over the query tape moves inmediately one cell to the right and can never move left.

In order to make a query to the oracle $A$, first $M$ writes a string in the query tape and enters a query state $q_{\text{query}}$. Then, automatically, $M$ changes its state to $q_{yes}$ if the query string belongs to the set $A$, or $q_{no}$ otherwise. When $M$ changes from the query state to one of the answer states, the query string is inmediately deleted and the head is positioned over the first cell from the left. This way, making a query to $A$ only takes one step of $M$.

The time complexity of $M^A$ is defined as follows. For any function $t:\nat\to\nat$ and any input $x$, the time complexity of $M^A$, denoted $time_{M^A}(n)$, is $t(|x|)$ if $M^A$ stops within $t(|x|)$ moves on input $x$. Note that the time it takes for $M^A$ to write in the query tape is counted towards its time complexity, whereas each oracle query counts only as one move.

Space complexity is defined similarly. For any function $s:\nat\to\nat$ and any input $x$, the space complexity of $M^A$, denoted $space_{M^A}(n)$, is $s(|x|)$ if the maximum number of work tape cells that $M^A$ uses at any time of its computation is at most $s(|x|)$ on input $x$; the number of query tape cells used are not counted.

A set $A$ is \emph{many-one reducible} to a set $B$, denoted $A\leq_m B$, if there exists a computable function $\phi$ such that $x\in A$ if and only if $\phi(x)\in B$. The set $A$ is \emph{Turing reducible} to $B$, denoted $\leq_T$, if there exists a TM $M$ such that $M^B$ computes $A$.


In this work we use $\leq_m^L$ and $\leq_T^L$ to denote many-one and Turing reductions bounded by $\mathcal{O}(\log n)$ space, respectively. Analogously, we use the superscript ``$poly$'' for polynomial-space reductions. We also use $\equiv$ to denote equivalence, that is, $A \equiv_T B$ if and only if $A\leq_T B$ and $B\leq_T A$.


\section{A Space-Hierarchy Theorem}\label{sec:hierarchy}

From Fact \ref{fac:algebraic} we know that any algebraic real number is computable in linear space. In this section, however, we show in Theorem \ref{the:hierarchy} that there exists an infinite hierarchy of computable real numbers. Thus, for trascendental numbers there is no such bound.

Before proving the main result of this section, we introduce the following technical lemma.

\begin{lemma}[App.\ref{app:output-transforms}]\label{lem:output-transforms}
Let $M$ be a Turing machine that computes a real number $x$ using $s(n)\geq \log n$ space. There exists a machine $M'$ that on input $0^n$ outputs the $n$-th digit of $x$ in space $\mathcal{O}(\max\{\log n, s(n)\})$.
\end{lemma}

\begin{theorem}[Space-Hierarchy Theorem for Real Numbers]\label{the:hierarchy}
Let $f$ and $g$ be two space-constructible functions where $f(n)=o(g(n))$ and $g(n)\geq f(n) \geq \log n$. Then $\dspace_\real(f(n)) \subsetneq \dspace_\real(g(n))$.
\end{theorem}
\begin{proof}
By the space hierarchy theorem of \cite{SHL65}, there exists a recursive set $\texttt{A}$ such that the set $\texttt{A}'=\{0^n \> |\> n\in \texttt{A}\}$ is computable in space ${\mathcal{O}(g(n))}$ but not computable in space ${\mathcal{O}(f(n))}$; in particular, there exists a machine $M_A$ that on input $0^n$ outputs 1 if $n\in A$ and outputs 0 if $n\notin A$ in space ${\mathcal{O}(g(n))}$, and there is no Turing machine with a unary input that decides $\texttt{A}$ in space ${\mathcal{O}(f(n))}$.

 We will construct a real number $a$ from $\texttt{A}$ such that $a$ is in $\dspace_\real(g(n))$ and not in $\dspace_\real(f(n))$. Let $a=0.a_1a_2\cdots$, and define $a_i=1$ if $i\in \texttt{A}$ and 0 otherwise.

For the sake of contradiction, assume that $a\in \dspace_\real(f(n))$, that is, there exists a machine $M_a$ that on input $0^n$ writes $0.a_1a_2\cdots a_n$ on its output tape and $space_{M_a}(n)=\mathcal{O}(f(n))$.

Now we construct a machine $N$ that on input $0^n$ decides if $n\in \texttt{A}$ using ${\mathcal{O}(f(n))}$ space, thus contradicting the fact that $\texttt{A}$ is not computable in space ${\mathcal{O}(f(n))}$.


Since $M_a$ computes $a$ in space $\mathcal{O}(f(n))$, from Lemma \ref{lem:output-transforms}, there exists a TM $M_a'$ that on input $0^n$ outputs the $n$-th bit of $a$ in space $\mathcal{O}(f(n))$. Thus, the machine $N$ simulates $M_a'$ on input $0^n$ and outputs whatever $M_a'$ outputs, and hence, the set $\texttt{A}$ is computable in space $\mathcal{O}(f(n))$, which is a contradiction.\qed
\end{proof}

\section{Space Complexity of Trascendental Numbers}\label{sec:trascendental}

From the Space-Hierarchy Theorem of the previous section it is understood that there is no space upper-bound on the set of trascendental numbers. There are trascendental numbers, however, with ``natural'' definitions that can be computed efficiently. 

\begin{theorem}\label{the:rec-space2}
	Let $f$ be any computable and strictly monotonically increasing function over the natural numbers.
	
	Let 
	\[
	\mu_f=\sum_{k=1}^\infty \frac{1}{10^{f(k)}}.
	\]
	Then the number $\mu_f$ is computable in space $\mathcal{O}(\max\{\log f(n), \log(n),s(\log n)\})$, where $s(n)$ is an upper bound on the space complexity of $f$.
\end{theorem}
\begin{proof}[sketch]
	We can construct a TM that outputs the representation of $\mu_f$ in base 10 as follows. For each input symbol (say, the $i$-th input symbol), we compute $f(i)$ and store it in the work tape, which needs $\mathcal{O}(s(n))$ and $\mathcal{O}(\log f(n))$ space, respectively. Then, we write a 1 in the $f(i)$-th cell of the output tape. Note that the function $f$ is strictly monotonically increasing. Thus, the output tape-head never goes back to the left. In order to execute the above procedure, we need to track the positions of the input and the output tapes, which needs $\mathcal{O}(\log n)$ and $\mathcal{O}(\log f(n))$ space, respectively. Thus, our algorithm uses $\mathcal{O}(\max\{\log f(n), \log(n), s(\log n)\})$ space.
	\qed
\end{proof}
See App.\ref{sec:rec-space2} for the detailed proof. Immediately we obtain the following corolary.

\begin{corollary}
Let $d$ be any positive integer. Any number of the form
\[
\sum_{k=1}^\infty \frac{1}{10^{k^d}}
\]
is computable in logspace.
\end{corollary}
In particular, the number $\mu_{n^2}$ is known to be trascendental, but $\mu_{n^3}$ is still open. Another interesting example is Liouvilles's constant which is when $f(k)=k!$, and hence, from  Theorem \ref{the:rec-space2} it follows that Liouville's constant is in $\dspace_\real(n\log n)$.

From Fact \ref{fac:rational} we know that finite automata or constant-space machines can only compute rational numbers. By a slight change in the definition of what it means for a number to be computable it allows finite automata to compute some irrational numbers.

Let $\Sigma_k$ be the set $\{0,\dots,k-1\}$ and let $w_r\cdots w_0$ in $\Sigma_k^r$ be the $k$-ary expansion of $n$, that is, $n=\sum_{i=0}^r w_i\cdot k^i$. A real number $a=a_0+\sum_{i\geq 1}a_ib^{-i}$ is $(k,b)$-automatic if there exists a finite state automaton that takes $n$ as input in radix-$k$ and outputs $a_n$. 

\begin{theorem}[App.\ref{sec:automatic}]\label{the:automatic}
Any $(k,b)$-automatic number is computable in logspace.
\end{theorem}

\section{Space-Bounded Real Numbers and Tally Sets}

Tally sets are languages over unary alphabets, that is, singleton alphabets.
In this section, we show in Theorem \ref{the:tally} that computable real numbers are computationally equivalent to tally sets, thus,  establishing a strong connection between space-bounded computational complexity and the theory of computable real numbers.

Before going into the main result of this section in Theorem \ref{the:tally}, first we present some technical lemmas which are also relevant for the remaining of this paper.

\begin{lemma}[Unary simulation lemma]\label{lem:unary-sim}
Let $M$ be a TM that computes a function $f:\{0\}^*\to \{0,1\}^*$ using space $s(n)$ with $s$ space-constructible. There exists a TM $N$ that simulates $M$ using space $\mathcal{O}(\max\{\log n,s(n)\})$ and computes a function $g:\{0,1\}^*\to \{0,1\}^*$ with $g(n)=f(0^n)$.
\end{lemma}
\begin{proof}
In order for $N$ to simulate $M$, we use a ``fake tape'' trick and exploit the fact that $M$ only works with unary inputs. The machine $N$ uses a  tape \texttt{pos}, initialized in 0, that stores in binary the position of the input head of $M$. Furthermore, $N$ will required an additional constant number of work tapes sufficient to run the simulation of $M$.

On input $x\in \{0,1\}^n$, first $N$ initializes \texttt{pos}:=0 and repeats the following procedure until $M$ stops its computation. If $\texttt{pos} \leq x$, simulate one step of $M$ with input symbol $0$; otherwise, simulate one step of $M$ with input symbol $\boxtimes$.  In either case, write in the output tape of $N$ whatever output symbol $M$ generates and update \texttt{pos} according to the move performed by $M$'s input tape head. Thus, the output of $N$ equals the output of $M$. 

The space used by the simulation of $M$ is $\mathcal{O}(s(n))$ and we only require $\mathcal{O}(\log n)$ bits to store the input tape position of $M$.\qed
\end{proof}

\begin{lemma}\label{lem:root}
For any $n\in \nat$ given in unary, $\lfloor\sqrt{n}\rfloor$ is computable in $\mathcal{O}(\log n)$ space.
\end{lemma}
\begin{proof}
We construct a TM $M$ that on input $0^n$ writes $\lfloor \sqrt{n}\rfloor$ on its output tape.
\begin{enumerate}
\item Initialize tapes $\texttt{floor}:=0$ and $\texttt{ceiling}:=1$ in binary.
\item Compute $n$ in binary and store it in $\texttt{number}$.
\item Enter loop.
	\begin{enumerate}
	\item[(a)] Compute $(\texttt{floor})^2$ and $(\texttt{ceiling})^2$ and store the results in \texttt{lower-bound} and \texttt{upper-bound}, respectively.
	\item[(b)] If $\texttt{lower-bound}\leq \texttt{number} < \texttt{upper-bound}$ then copy the contents of \texttt{floor} to the output tape and stop;
	\item[(c)] else increment the contents of \texttt{floor} and \texttt{ceiling} by 1.
	\end{enumerate} 
\end{enumerate}

The procedure above works by checking in each step of the loop the relation $\texttt{floor}^2 \leq n < \texttt{ceiling}^2$, and since $n$ is finite the loop terminates in finite time. Furthermore, tapes \texttt{floor}, \texttt{ceiling}, \texttt{number}, \texttt{lower-bound} and \texttt{upper-bound} use $\mathcal{O}(\log n)$ bits of storage. \qed
\end{proof}

For any nonnegative integers $i$ and $j$, define a pairing function $\langle i,j\rangle$ as a bijective function from $\nat^2$ to $\nat$ given by $\langle i,j\rangle =(i+j+1)(i+1)/2+y$. A pairing function $\langle i,j,k\rangle$ from $\nat^3$ to $\nat$ is easily defined inductively as $\langle \langle i,j\rangle,k\rangle$.

\begin{lemma}\label{lem:pairing}
For any nonnegative integers $i,j$ of at most $\mathcal{O}(n)$ bits, the pairing function $\langle i,j\rangle$ and its inverse are computable in $\mathcal{O}(\log n)$ space.
\end{lemma}
\begin{proof}
The computation of $\langle i,j\rangle$ is clearly computable in $\mathcal{O}(\log n)$ space, because we only need to do four additions, one multiplication and one bit shift which can all be done using logspace.
 
If $h=\langle i,j\rangle$, we can obtain $i$ and $j$ using the following procedure. Let $\Delta(x)=x(x+1)/2$, then
\begin{align}
c	&=\lfloor \sqrt{2h}  -1/2\rfloor,\label{eq:floor}\\
i	&=h-\Delta(c),\label{eq:i}\\
j	&=c-i+2.\label{eq:j}
\end{align}
The arithmetic operations of multiplication and addition of Eq.(\ref{eq:i}) and Eq.(\ref{eq:j}) can be done in logspace. For the square root of Eq.(\ref{eq:floor}), however, all known algorithms use linear space. For the inversion of the pairing function, however, we only need the floor which can be done in logspace by Lemma \ref{lem:root}.\qed
\end{proof}

\begin{theorem}\label{the:tally}
Let $s:\nat\to\nat$ be any space-constructible function such that $s(n)\geq \log n$.  A real number $x$ is in $\dspace_\real(s(n))$ if and only if the tally set
\[
T_{\phi_x}=\{0^{\langle n,i,b\rangle } \> | \> n,i\in \nat, 1\leq i \leq n, b\in \{0,1\},\phi_{x}(n)_i=b\}
\]
is computable in space $\mathcal{O}(s(n))$, where $\phi_x(n)_i$ denotes the $i$-th bit in a base-2 representation of $\phi_x(n) \in CF_x$.
\end{theorem}
\begin{proof}
We construct a machine $M$ computing $T_{\phi_x}$ that operates as follows. Given an input $w\in \{0\}^*$, first $M$ inverts the pairing function using the length of $w$ as a parameter to obtain $n,i$ and $b$ in binary. If $\phi_x(n)_i=b$ then $M$ accepts $w$, otherwise it rejects $w$. Since $\phi_x$ is computable in space $s(n)$, we only need to show that all the other steps can be done in $\mathcal{O}(s(n))$ space. 

In order to invert the pairing function $\langle i,j,k\rangle$ we need to invoque Lemma \ref{lem:pairing} two times, first to obtain $k$ and $h=\langle i,j\rangle$, and a second time to obtain $i$ and $j$ from $h$, using  $\mathcal{O}(\log\log n)$ space. This proves the first part of the implication.

Now suppose that $T_{\phi_x}$ is computable in space $\mathcal{O}(s(n))$, and we want to prove that $x$ is computable in space $\mathcal{O}(s(n))$.  By Lemma \ref{lem:unary-sim}, there exists a machine $M$ that computes using $\mathcal{O}(s(n))$ space a function $f_M:\{0,1\}^n\to \{0,1\}$ defined as $f_M(y)=1$ if $0^n \in T_{\phi_x}$ and $f_M(y)=0$ otherwise, where $n$ is the integer number represented by $y$.

A machine $N$ for computing $x$ works as follows. On input $0^n$, in a tape \texttt{post} initialized in 0, keep a count in binary of the output tape position of $N$, and in a tape \texttt{length} store $n$ in binary. Then using a tape \texttt{counter} initialized in 0, repeat the following procedure. Simulate $M$ with \texttt{counter} as input to obtain its single-symbol output, say $z$. If $z=1$ then $\texttt{counter}=\langle i,j,k\rangle$ for some nonnegative integers $i,j,k$ and by Lemma \ref{lem:pairing} we can obtain $i,j$ and $k$ in logspace. If $i=\texttt{length}$ and $j=\texttt{pos}$, then write $k$ in the output tape of $N$, increment \texttt{pos} in one and set \texttt{counter}:=0; otherwise, if $i\neq\texttt{length}$ or $j\neq\texttt{pos}$, increment \texttt{counter} in one. Repeat this procedure to obtain the first $n$ output bits of $x$. Since $M$ runs in $\mathcal{O}(s(n))$ space and $\texttt{counter}$ only stores $\mathcal{O}(\log n)$ bits, then $N$ runs in $\mathcal{O}(s(n))$ space.\qed
\end{proof}

The proposition below presents an application of Theorem \ref{the:tally}.

\begin{proposition}\label{pro:set}
Let $\texttt{A}$ be any subset of $\nat$ and let $\texttt{A}'=\{0^n\> | \> n\in \texttt{A}\}$. If $\texttt{A}'$ is decidable in $\mathcal{O}(\log n)$ space, then $\mu=\sum_{p\in \texttt{A}}(1/10)^p$ is in $\dspace_\real(\log n)$.
\end{proposition}

Proposition \ref{pro:set} follows immediately from Theorem \ref{the:tally} and Lemma \ref{lem:tally-set}.
\begin{lemma}[App.\ref{sec:tally-set}]\label{lem:tally-set}
$\texttt{T}_{\phi_\mu} \leq_m^L \texttt{A}'$.
\end{lemma}

If we let $\texttt{A}=\texttt{PRIMES}$, the set of prime numbers, we have that $\mu$ has a 1 in every prime position and it is clearly not rational. From Proposition \ref{pro:set} it follows that $\mu$ is computable in logspace because the tally set $\texttt{PRIMES}'$ is computable in logspace; we do not know, however, if $\mu$ is algebraic or trascendental for $\texttt{A}=\texttt{PRIMES}$.

\section{Space-Bounded Real Numbers and Non-Tally Sets}
In the previous section we presented a characterization between tally sets and space-bounded real numbers. In this section we explore relations between space-bounded real numbers and languages whose alphabets are not singletons. Languages representing real numbers were introduced by Ko \cite{Ko91}. Given $x\in \real$, for each $\phi_x \in CF_x$ define the language $\texttt{L}_{\phi_x}$ as the set of strings $s$ representing dyadic rational numbers $d$ such that $d\leq \phi(n)$ with $n=prec(d)$. 

\begin{theorem}\label{the:language}
Let $s:\nat\to\nat$ be any space-constructible function with $s(n)\geq \log n$. If $x\in \dspace_\real(s(n))$, then $\texttt{L}_{\phi_x}$ is computable in space $O(s(n))$ for some $\phi_x\in CF_x$. 
\end{theorem}
\begin{proof}
Let $M_{\phi_x}$ be TM that computes $\phi_x$ in space $s(n)$ for some $\phi_x \in CF_x$. We construct another machine $N$ that on input $d$ simulates $M_{\phi_x}$ and checks if $d$ is less than or equal to the output generated by $M_\phi$.

To make the proof work we need two technical considerations. First, $N$ must compares its input $d$ against the output $\gamma$ of $M_{\phi_x}$ using space $\mathcal{O}(s(n))$. To compare $d$ against $\gamma$ using space $\mathcal{O}(s(n))$ it suffices to remember only one symbol of $\gamma$ at a time. Second, in order to avoid simulating $M_{\phi_x}$ using linear space, we make use of the ``fake tape'' trick of Lemma \ref{lem:unary-sim} for $M_{\phi_x}$, that is, we store the position of the input head of $M_{\phi_x}$ and simulate $M_{\phi_x}$ by feeding it one 0 symbol at a time.  Note that here we cannot use Lemma \ref{lem:unary-sim} directly, because we need to remember the output of $M_{\phi_x}$ and compare it against the input $d$ which would use linear space. 

The machine $N$ has one tape to simulate $M_\phi$, another tape \texttt{pos} to record the position of the input head of $M_\phi$, and a tape $\texttt{prec}$ to store the precision of the input.

With no loss of generality, we consider dyadic numbers betweenen 0 and 1. Let $d=0.d_{1}\cdots d_{n}$ be an input for $N$, and let $\gamma=0.\gamma_{1}\cdots \gamma_{n}$ be the output of $M_\phi$ on input $0^n$.

The algorithm executed by $N$ is the following.
\begin{enumerate}
\item Compute $prec(d)$ and store it in the tape $\texttt{prec}$ in binary. To do this, first initialize $\texttt{prec}:=0$ and scan each input symbol after the binary point. For each input symbol after the binary point, increment \texttt{prec} by one until a blank symbol tape $\boxtimes$ is encountered.
\item Initialize $\texttt{pos}:=0$ and position the input head of $N$ on the first symbol after the binary point.
\item Enter a loop.
	\begin{enumerate}
	\item If $\texttt{pos}\leq \texttt{prec}$, simulate one step of $M_\phi$ with input symbol $0$; else, simulate one step of $M_\phi$ with input symbol $\boxtimes$.
	\item Update the contents of tape \texttt{pos} accordingly, that is, if $M_\phi$ moved its input head right or left, increment or decrease \texttt{pos} by one, respectively.
	\item If $M_\phi$ did not generated any output symbol, then do nothing.
	\item Suppose $M_\phi$ generated an output symbol $\gamma_i$ and let $j$ be the current position of the input head of $N$. If $\gamma_i< d_j$, accept; if $\gamma_i> d_j$, reject.
	\item If $M_\phi$ entered a halting state, then stop the computation and accept.
	\end{enumerate}
\end{enumerate}

Tapes \texttt{pos} and \texttt{prec} use at most $\mathcal{O}(\log n)$ bits and the simulation of $M_{\phi_x}$ uses $\mathcal{O}(s(n))$ bits.\qed
\end{proof}


\begin{proposition}
$\texttt{L}_{\phi_x} \leq_{T}^{L} \texttt{T}_{\phi_x}$.
\end{proposition}
\begin{proof}
Let $d=0.d_1\cdots d_n \in \texttt{L}_{\phi_x}$. We construct a machine that decides $d$ with oracle access to $\texttt{T}_{\phi_x}$ using $\mathcal{O}(\log n)$ space.

We use tapes \texttt{length} and \texttt{pos} to store in binary the length $n$ of the input and the position of the input head, respectively. For each $i$ with $1\leq i \leq n$ and using tapes \texttt{length} and \texttt{pos}, by Lemma \ref{lem:pairing}, the pairing functions $\langle n,i,0\rangle$ and $\langle n,i,1\rangle$ are computable in $\mathcal{O}(\log\log n)$ space. Note that $\langle n,i,0\rangle$ and $\langle n,i,1\rangle$ have $\mathcal{O}(\log n)$ bits. Then query $0^{\langle n,i,0\rangle}$ and $0^{\langle n,i,1\rangle}$  to the oracle $\texttt{T}_{\phi_x}$ to determine the correct bit $b$. If $b>d_i$, reject; if $b<d_i$ accept. After checkin all bits of $\phi_x(n)$, accept.
\qed
\end{proof}

\begin{proposition}
$\texttt{T}_{\phi_x}   \leq_T^{poly} \texttt{L}_{\phi_x}$.
\end{proposition}
\begin{proof}
Suppose with no loss of generality that $0<x<1$. We construct a machine $M$ that, on input $0^n$ and using $\texttt{L}_{\phi_x}$ as oracle, decides $\texttt{T}_{\phi_x}$ using $\mathcal{O}(\log n)$ space.

In a tape called \texttt{length} we store $n$ in binary using $\mathcal{O}(\log n)$ bits. By Lemma \ref{lem:pairing}, we can invert the pairing function $n=\langle m,i,b\rangle$ using $\mathcal{O}(\log\log n)$ space. Note that each $m$ and $i$ also have $\mathcal{O}(\log n)$ bits. 

The reduction implements the following procedure. Use exhaustive search to find a largest dyadic number $d=0.d_1\cdots d_i\cdots d_m$ such that $d\in L_{\phi_x}$. If $d_i=b$ accept $0^n$, otherwise reject.

Suppose that $0^n\in T_{\phi_x}$. Then, using the procedure above we obtain a dyadic number $d$ whose first $m$ digits agree with $x$. Therefore, $d_i=b$ and $0^n$ is accepted. Now suppose that $0^n \notin T_{\phi_x}$. Since $0^n\notin T_{\phi_x}$ we have that $\phi(m)_i\neq b$. Hence, $d_i\neq b$ and $0^n$ is rejected.

To finish the proof, we need to show that the above mentioned procedure can be implemented in polynomial space. The tape \texttt{length} is used as input to invert the pairing function $n=\langle m,i,b\rangle$, which suffices with $\mathcal{O}(\log n)$ bits. During the exhaustive search procedure, we only need to remember $m=\mathcal{O}(\log n)$ bits, which again there are at most $\mathcal{O}(n^c)$.\qed
\end{proof}

\begin{corollary}
$\texttt{T}_{\phi_x}   \equiv_T^{poly} \texttt{L}_{\phi_x}$.
\end{corollary}

\section{Non-uniform Deterministic Finite Automata}\label{sec:automata}
Constant-space machines are modeled by finite automata, and from Fact \ref{fac:rational} it is clear that no irrational number is computable by constant-space machines. As a final result of this paper, we show that constant-space machines cannot recognize irrational numbers even in the presence of external aid which we model as advice.

A deterministic finite automaton with advice has a read-only advice tape in which an advice string is written prior to the computation. The advice string is allowed to depend on the length of the input, but not on the input itself. 

A deterministic finite automaton with advice is formally defined as a 7-tuple $M=(Q, \Sigma, \Gamma, \delta, q_1, F, \hat{a})$, where $Q$ is a finite set of states, $\Sigma$ (resp. $\Gamma$) is a finite input (resp. output) alphabet, $\delta: Q\times (\Sigma\cup\{\mbox{\textcent}, \$\})^2 \longrightarrow Q\times \{L,R\}^2 \times \Gamma\cup\{\varepsilon\}$ is a state transition function, $q_1$ is the initial state, $F\subseteq Q$ is the set of halting states, and $\hat{a}=\{a_0, a_1, \ldots\}$ is a set of advice strings. The input $w$ is written in the input tape as $\mbox{\textcent}w\$$, and the advice string is written in the advice tape as $\mbox{\textcent}a_{|w|}\$$, where $|w|$ represents the length of the input $w$, and \textcent (resp. \$) is the left (resp. the right) end-marker. In the initial configuration, the input tape head and the advice tape head scan the left end-markers, and the state is in $q_1$. Then, at each step of the computation, the automaton changes its state, moves the input and the advice tape heads by one cell, and outputs a symbol (which can be an empty word $\varepsilon$) according to the state transition function. When the automaton reaches one of the halting states, it halts.	

Note that our definition of deterministic finite automata with advice can be seen as the Mealy machine equipped with an advice tape.

We say that a deterministic finite automaton computes a real number $x$ if for any unary input $0^n$, it outputs $\phi_x(n)\in CF_x$.

A {\it (complete) configuration} of a deterministic finite automaton with advice is represented as a triple $(q, h_1, h_2)$, where $q$ is the current state, $h_1$ is the position of the input tape head, and $h_2$ is the position of the advice tape head. We define a {\it partial configuration} as a triple $(q, a, h)$ where $q$ is the current state, $a$ is the symbol scanned by the input tape head, and $h$ is the position of the advice tape head.

\begin{lemma}
\label{lemma:form_of_output}
Let an input string be $0^n$. We consider the computation of a deterministic finite automaton with constant-sized advice that outputs more than $n$ symbols. Then, the output string is of the form  $w_1 w^m_2 w_3$ where $|w_1|, |w_2|\in O(1)$, and $m\in \omega(1)$.
\end{lemma}
\begin{proof}
Let $c_1, c_2, \ldots $ be the sequence of configurations made in a computation of the automaton for input $0^n$. Let $S = c_{l_1}, c_{l_2},\ldots$ be a subsequence induced by all the configurations in which the tape head scans the left or the right endmarker. Let $l_{adv}$ be the length of an advice string and let $Q$ be the set of states of the automaton. We consider the sequence of the first $2\cdot |Q|\cdot l_{adv} + 1$ configurations in $S$, $c_{l_1}, \ldots, c_{l_{2\cdot |Q|\cdot l_{adv}+1}}$. We have the following two cases.
\vspace{0.2cm}

\noindent{\bf Case 1. There exists $i$ ($1\leq i \leq 2\cdot |Q|\cdot l_{adv}+1$) such that the automaton outputs $\omega(1)$ symbols along the transitions from $c_{l_i}$ to $c_{l_{i+1}}$.}

Let $i$ be a smallest number that satisfies the condition of Case 1. We consider a sequence of configurations, $c_{l_i},c_{l_i + 1}, \ldots, c_{l_{i+1}}$. Let $\hat{c}_{l_i}, \hat{c}_{l_i + 1}, \ldots, \hat{c}_{l_{i+1}}$ be the sequence of partial configurations that corresponds to $c_{l_i}, c_{l_i + 1}\ldots, c_{l_{i+1}}$. Then, there exists $s$ and $t$ ($l_i < s<t < l_{i+1}$) such that $\hat{c}_s = \hat{c}_t$ since the number of possible partial configurations is a constant while the length of the sequence, $\hat{c}_{l_i}, \hat{c}_{l_i + 1}, \ldots, \hat{c}_{l_{i+1}}$, is $\omega(1)$. We consider the smallest $s$ and $t$. Then, $s-l_i, t-l_i \in O(1)$. Note that for the partial configurations,  $\hat{c}_{l_i + 1}, \hat{c}_{l_i + 2}, \ldots, \hat{c}_{l_{i+1}-1}$, the scanned input symbol is always $0$. Therefore, the transitions of partial configurations from the initial configuration can be written as $\hat{c}_1, \ldots, \hat{c}_{l_i}, \ldots, \hat{c}_s, \ldots, \hat{c}_t (=\hat{c}_s) , \ldots, \hat{c}_t, \ldots, \hat{c}_{l_{i+1}}, \ldots$, i.e. once the automaton reaches the configuration $c_s$, it repeats the transitions from $\hat{c}_s$ to $\hat{c}_t$ until it reads the left or the right end-marker in $c_{l_{i+1}}$. This means that the automaton outputs $w_1 w^m_2 w_3$ where $w_1$ is generated by the transitions from $\hat{c}_1$ to $\hat{c}_s$, $w^m_2$ is generated by the iterations of the transitions from $\hat{c}_s$ to $\hat{c}_t$, and $w_3$ is generated from the rest of the transitions that include the transitions after $c_{l_{i+1}}$. Note that $i$ is bounded above by a constant ($2\cdot |Q|\cdot l_{adv}+1$). Also, note that for each $j(<i)$, the automaton outputs $O(1)$ symbols along the transitions from $c_{l_j}$ to $c_{l_{j+1}}$. As mentioned above, $s-l_i\in O(1)$. Thus, $|w_1|\in O(1)$.  Also, $|w_2|\in O(1)$ since $t-s \in O(1)$. Recall that by the condition of Case 1, the automaton outputs $\omega(1)$ symbols along the transitions from $c_{l_i}$ to $c_{l_{i+1}}$. Thus, $m\in \omega(1)$.
\vspace{0.2cm}

\noindent{\bf Case 2. For all $i$ ($1\leq i \leq 2\cdot |Q|\cdot l_{adv}+1$), the automaton outputs $O(1)$ symbols along the transitions from $c_{l_i}$ to $c_{l_{i+1}}$.}

In this case, there exists $i$ and $j$ ($1\leq i< j\leq 2\cdot |Q|\cdot l_{adv}$) such that $c_{l_i}=c_{l_j}$ since the number of possible configurations in which the automaton scans the left or the right endmarker is $2\cdot |Q|\cdot l_{adv}$. Thus, the automaton repeats the transitions from $c_{l_i}$ to $c_{l_j}(=c_{l_i})$ infinitely. Therefore, the transitions from the initial configuration can be written as  $c_1,\ldots,c_{l_i},\ldots,c_{l_j}(=c_{l_i}),\ldots, c_{l_j}(=c_{l_i}), \ldots$, which means that the automaton outputs $w_1 w^m_2$ where $w_1$ is generated by the transitions from $c_1$ to $c_{l_i}$, and $w^m_2$ is generated by the iterations of the transitions from $c_{l_i}$ to $c_{l_j}$. Note that $|w_1|, |w_2| \in O(1)$ by the condition of Case 2.
\qed
\end{proof}

\begin{theorem}
\label{theorem:impossibility_constant_advice_DFA}
If a deterministic finite automaton with constant-sized advice computes a real number $x$, then $x$ is represented as $w_1 w_2^m$, where $m$ can be infinite.
\end{theorem}
\begin{proof} 
We consider a deterministic finite automaton with constant-sized advice that computes a real number $x$. Let $O_n$ denote the output for the input $I_n=0^n$. Note that by Lemma~\ref{lemma:form_of_output}, we can assume that the output of the automaton is of the form $w_1 w_2^m w_3$ where $|w_1|, |w_2|\in O(1)$, and $m\in \omega(1)$. Then, for any input $I_n$, there exists an input $I_{\hat{n}}=0^{\hat{n}} (\hat{n}>n)$ such that for $I_{\hat{n}}$, the automaton outputs $\hat{w}_1 \hat{w}_2^{\hat{m}} \hat{w}_3$ where $|\hat{w}_1|, |\hat{w}_2|\in O(1)$, and $|\hat{w}_1 \hat{w}_2^{\hat{m}}|> n$. This means that the first $n$ symbols of $O_n (=\phi_x(n))$ can be written as $\hat{w}_1 \hat{w}_2^{\hat{m}'} \hat{w}'_2$ for some $\hat{m}'< \hat{m}$ where $\hat{w}'_2$ is a prefix of $\hat{w}_2$ since the first $n$ symbols of $\phi_x(n')$ to the right of the binary point agree with $\phi_x(n)$. Therefore, $x$ is represented as $w_1 w_2^m$ for some $m$, where $m$ can be infinite.
\qed
\end{proof}

Theorem~\ref{theorem:impossibility_constant_advice_DFA} implies that deterministic finite automata with constant-sized advice cannot compute any irrational numbers.

~\\
\noindent\textbf{Acknowledgements.}  The authors thank Abuzer Yakaryilmaz for useful discussions.

\bibliographystyle{splncs03}
\bibliography{../../../library}

\appendix
\section*{Appendix}

\section{Proof of Fact \ref{fac:algebraic}}\label{sec:algebraic}
%
%
Since computing the integer parts are trivial,  we pick a real algebraic number, say $ \alpha $, between 0 and 1. Let $f(x)=a_rx^r+\cdots +a_1x+a_0$ be its minimal polynomial. The rational number $ \alpha_j $ represents the first $ j $-th bit(s) of $ \alpha $ after the decimal point. We pick a sufficiently big integer $ k $ such that $ \alpha $ is the closest root to $ \alpha_k $. Since $ f(x) $ is minimal either $ f(\alpha_k) \leq 0 $ and $ f(\alpha_k+2^{-k}) > 0 $ or vice versa. We can easily determine the case by computing both values. We assume  the first case. (We can take $ -f(x) $ for the other case.)


Let $M$ be a Turing machine that on input $0^n$ operates as follows. The first $k$ bits of $\alpha$ are stored in the description of $M$ as $\alpha_k$. If $n\leq k$, the machine $M$ outputs each bit of $\alpha_k$ up to $n$. In this case, $M$ does not use any work tape and constant space suffices to output the first $k$ bits.

If $n>k$ the machine $M$ operates as follows. First, write $\alpha_k$ in the output tape. Let $t_0,t_1,\dots,t_r,t_{r+1}$ be $r+2$ different work tapes in $M$. Copy $\alpha_k$ on tape $t_0$. Starting from $i=k$ repeat the procedure below $n-k$ times.
\begin{enumerate}
\item For each $j=1,2,\dots,r$ compute $(t_0+2^{-(i+1)})^j$ and store it on tape $t_j$.
\item For each $j=1,2,\dots,r$ multiply $a_j$ with $t_j$ and store the result in tape $t_j$.
\item Add $t_r+\cdots + t_1+a_0$ and store the result in tape $t_{r+1}$.
\item If $t_{r+1}\leq 0$, write 1 on the output tape and copy it on tape $t_0$; this will make $\alpha_{i+1}=\alpha_i+2^{-(i+1)}$
\item If $t_{r+1}>0$, write 0 on the output tape; this will make $\alpha_{i+1}=\alpha_i$.
\item Increment $i$ in one and if $i\leq n$ return to step 1.
\end{enumerate}
Each addition and multiplication can be done in logspace. Each work tape $t_j$, however, stores at most $\mathcal{O}(n)$ bits; this is because in the last iteration of the algorithm the tape $t_0$ stores $\alpha_n$.

\section{Proof of Theorem \ref{the:rec-space2}}\label{sec:rec-space2}
We construct a machine $M$ that outputs the representation of $\mu_f$ in base 10 as follows. The TM $M$ has three main tapes called \texttt{pos\_in}, \texttt{pos\_out} and  \texttt{res}. The \texttt{pos\_in} tape will track the position of the input tape-head. The \texttt{pos\_out} tape will track the position of the output tape-head. The \texttt{res} tape will store temporary results for the computation of $f$. Besides these three main tapes, the machine $M$ will use some extra finite number of tapes that depends on the computation of $f$. All tapes involved will use a binary alphabet.

The algorithm is the following.
\begin{enumerate}
	\item If the first symbol under the input tape-head is $\boxtimes$, then write 0 in the output and stop the computation. 
	\item Initialize tapes. Write 1 in \texttt{pos\_in} and \texttt{pos\_out}. Leave the others blank.
	\item Write "0." in the output tape.
	\item Enter a cycle\label{alg:liou-cycle1}
	\begin{enumerate}
		\item[(a)] If the input tape-head scans $\boxtimes$ (i.e., the input tape-head reaches the end of the input), break the loop.
		\item[(b)] Compute $f(\texttt{pos\_in})$ and store the result in \texttt{res}.
		\item[(c)] Enter cycle
		\begin{itemize}
			\item[(c.1)] If \texttt{pos\_out}=\texttt{res}, write 1 in the output tape and break the inner loop.
			\item[(c.2)] If \texttt{pos\_out} $\neq$ \texttt{res} (in this case, \texttt{pos\_out} $<$ \texttt{res} ), write a 0 in the output tape and move the output tape-head to the right by one cell. At the same time, increment \texttt{pos\_out}.
		\end{itemize}
		\item[(d)] Move the input tape-head to the right, and increment \texttt{pos\_in}. 
	\end{enumerate}
\end{enumerate}

At the start of iteration $i$ of the cycle of step \ref{alg:liou-cycle1}(c), \texttt{pos\_in}=$i$. The algorithm checks if the current tape-head position of the output tape (\texttt{pos\_out}) equals the value of \texttt{res}. If $\texttt{pos}$ is not equal to $\texttt{res}$, the algorithm moves the output tape head by one cell towards \texttt{res} (i.e., to the right) and repeats the cycle. When $\texttt{res}=\texttt{pos\_out}$, the algorithm writes a 1 in the output tape and breaks the inner circle. 

We analyze the space complexity. The tape $\texttt{pos\_in}$ uses at most $\mathcal{O}(\log n)$ cells. In the worst case, the tape $\texttt{pos\_out}$ and $\texttt{res}$ store the value of $f(n)$, thus, it uses at most $\mathcal{O}(\log f(n))$ cells. The amount of space to perform the computation of $f$ is bounded from above by $\mathcal{O}(s(\log n))$. Thus, it follows that the total amount of space is $\mathcal{O}(\max\{\log f(n), \log(n),s(\log n)\})$.
\qed

\section{Space Complexity of Arithmetic Operations on Rational numbers}

\begin{lemma}[\cite{CDL01}]
Given any two integer numbers $a$ and $b$ in binary, the division of $a$ by $b$ can be done in logspace.
\end{lemma}

\begin{lemma}
Any two rational numbers can be multiplied in logspace.
\end{lemma}
\begin{proof}
Let $r_1=a_1/b_1$ and $r_2=a_2/b_2$ be two rational numbers. We assume that $a_1,b_1,a_2,b_2$ are given as inputs to a Turing machine $M$ in radix-2.

The machine $M$ multiplies $a_1$ and $a_2$ and then $b_1$ and $b_2$, which can be done in logspace. Let $r_3=a_3/b_3$ where $a_3=a_1\cdot a_2$ and $b_3=b_1\cdot b_2$.\qed
\end{proof}

\begin{lemma}
Any two rational numbers can be added in logspace.
\end{lemma}
\begin{proof}
Let $r_1=a_1/b_1$ and $r_2=a_2/b_2$ be two rational numbers. We assume that $a_1,b_1,a_2,b_2$ are given as inputs to a Turing machine $M$ in radix-2.

The machine $M$ operates as follows. First, $M$ multiplies in logspace $b_3=b_1 \cdot b_2$. Then $M$ computes $a_3=a_1\cdot b_2 + a_2 \cdot b_1$, where each addition and multiplication takes logspace.\qed
\end{proof}

\begin{lemma}\label{lem:rat-reduce}
Any fraction can be transformed to its lowest terms in linear space.
\end{lemma}
\begin{proof}
Let $r=a/b$ be a rational number. With no loss of generality assume that $a\geq b$ and that $a$ and $b$ are given as inputs in radix-2. Then execute the following algorithm.

\begin{enumerate}
\item Let $a'\gets a$, $b' \gets b$ and $i\gets 2$.
\item Repeat until $b'<i$.
	\begin{enumerate}
	\item Divide $a'/i$ and $b'/i$.
	\item If both divisions have residue 0, let $a'\gets a'/i$ and $b'\gets b'/i$; else, increment $i$ in one.
	\end{enumerate}
\end{enumerate}

The procedure above outputs $r$ to its lowest terms and it takes linear space because we only need to keep $a'$ and $b'$ in memory at all times and division also takes logspace.\qed
\end{proof}

\begin{corollary}
Any two rational numbers can be added and multiplied to its lowest terms in linear space.
\end{corollary}

\section{Proof of Lemma \ref{lem:output-transforms}}\label{app:output-transforms}

We construct a machine $M'$ that simulates $M$. Let \texttt{counter} and \texttt{max} be two working tapes. First, $M'$ sets \texttt{counter}:=0 and counts in binary the input size using \texttt{max}. Thus, \texttt{max} stores the input length in binary. Then, $M'$ sets its input tape head at the beginning of its input tape and starts simulating $M$ using its own input tape as input for $M$. Every time $M$ generates an output symbol, if $\texttt{counter}\neq \texttt{max}$ then increment \texttt{counter}; otherwise, if  $\texttt{counter}= \texttt{max}$, let $M'$ write the output symbol of $M$ on $M'$'s output tape and stop the computation.

The space utilized by $M'$ is $\mathcal{O}(\log n)$ because $\texttt{max}$ and $\texttt{counter}$ tapes use at most $\mathcal{O}(\log n)$ space, and the simulation of $M$ uses $\mathcal{O}(s(n))$ space.\qed

\section{Proof of Lemma \ref{lem:tally-set}}\label{sec:tally-set}

We construct a logspace machine $M$ with oracle access to the set $\texttt{A}'$. On input $0^n$, using Lemma \ref{lem:pairing} compute in logspace the inverse of the pairing function $n=\langle m,i,b\rangle$ to obtain $m,i,b$. Note that each $m,i,b$ has $\mathcal{O}(\log n)$ bits. If $b=1$ and $i\in \texttt{A}'$ then accept, or if $b=0$ and $i\notin \texttt{A}'$ then also accept; in other cases, reject. The reduction is computable in logspace because of Lemma \ref{lem:pairing} and the logspace computability of $\texttt{A}'$.\qed

\section{Proof of Theorem \ref{the:automatic}}\label{sec:automatic}
Let $a=a_0+\sum_{i\geq 1}a_ib^{-i}$ be a $(k,b)$-automatic real number and let $M_a$ be a finite-state automaton that computes $a$. We construct a Turing machine $M$ that on input $0^n$ outputs $a_1\cdots a_n$.

With no loss of generality, assume that $a_0=0$, that is, $a$ is in the interval $[0,1)$. Let \texttt{count} be tape that keeps a counter in binary. If the input tape is empty, write 0 in the output tape and stop. If the input tape is not emtpy, then for each symbol in the input tape, increment \texttt{count} by one and simulate $M_a$ using the contents of \texttt{count} as input. Thus, for each input symbol that is scanned, write the output of $M_a$ on the output tape of $M$.

\end{document}